\documentclass[12pt,twoside]{article}

 \usepackage{float}
\usepackage{graphicx}
\usepackage{epstopdf}
\usepackage{graphicx}
\usepackage{epic}
\usepackage{multirow}
\usepackage{pst-poly}  
\usepackage{pst-plot}  
\usepackage{pst-poly}  
\usepackage{tikz}
\usepackage{threeparttable}
\usepackage{xcolor}
\usetikzlibrary{arrows,shapes,chains}

\renewcommand{\paragraph}{\roman{paragraph}}
\usepackage[a4paper]{geometry}
\makeatletter
\renewcommand\title[1]{\gdef\@title{\reset@font\Large\bfseries #1}}
\renewcommand\section{\@startsection {section}{1}{\z@}%
                                   {-3.5ex \@plus -1ex \@minus -.2ex}%
                                   {2.3ex \@plus.2ex}%
                                   {\normalfont\large\bfseries}}
\renewcommand\subsection{\@startsection{subsection}{2}{\z@}%
                                     {-3ex\@plus -1ex \@minus -.2ex}%
                                     {1.5ex \@plus .2ex}%
                                     {\normalfont\normalsize\bfseries}}
\renewcommand\subsubsection{\@startsection{subsubsection}{3}{\z@}%
                                     {-2.5ex\@plus -1ex \@minus -.2ex}%
                                     {1.5ex \@plus .2ex}%
                                     {\normalfont\normalsize\bfseries}}

\def\@runningauthor{}\newcommand{\runningauthor}[1]{\def\runningauthor{#1}}
\def\@runningtitle{}\newcommand{\runningtitle}[1]{\def\runningtitle{#1}}

\renewcommand{\ps@plain}{%
\renewcommand{\@evenhead}{\footnotesize\scshape \hfill\runningauthor\hfill}
\renewcommand{\@oddhead}{\footnotesize\scshape \hfill\runningtitle\hfill}}

\newcommand{\F}{\mathbb{F}}

\newcommand {\dd} {{\mathbf{{\rm d}}}}
\g@addto@macro\bfseries{\boldmath}

\makeatother

\usepackage{amsthm,amsmath,amssymb}
\usepackage{cite}
\usepackage{graphicx}

\usepackage[colorlinks=true,citecolor=black,linkcolor=black,urlcolor=blue]{hyperref}

\theoremstyle{plain}
\newtheorem{theorem}{Theorem}[section]

\newtheorem{lem}[theorem]{Lemma}
\newtheorem{cor}[theorem]{Corollary}

\theoremstyle{definition}

\newtheorem{example}[theorem]{Example}

\theoremstyle{remark}
\newtheorem{remark}[theorem]{Remark}

\runningauthor{}

\date{}

\begin{document}

  \title{The weight enumerator polynomials of the lifted codes of the projective Solomon-Stiffler codes
  }
\author{ Minjia Shi, Shitao Li, Tor Helleseth\thanks{Minjia Shi and Shitao Li are with the School of Mathematical Sciences, Anhui University, Hefei, China (email: smjwcl.good@163.com, lishitao0216@163.com).
Tor Helleseth is with the Department of Informatics, University of Bergen, Bergen, Norway (email: tor.helleseth@uib.no).}}

\date{}
\maketitle

\begin{abstract}
Determining the weight distribution of a code is an old and fundamental topic in coding theory that has been thoroughly studied. In 1977, Helleseth, Kl{\o}ve, and Mykkeltveit presented a weight enumerator polynomial of the lifted code over $\F_{q^\ell}$ of a $q$-ary linear code with significant combinatorial properties, which can determine the support weight distribution of this linear code. The Solomon-Stiffler codes are a family of famous Griesmer codes, which were proposed by Solomon and Stiffler in 1965. In this paper, we determine the weight enumerator polynomials of the lifted codes of the projective Solomon-Stiffler codes using some combinatorial properties of subspaces. As a result, we determine the support weight distributions of the projective Solomon-Stiffler codes. In particular, we determine the weight  hierarchies of the projective Solomon-Stiffler codes. 
\end{abstract}

\section{Introduction}
Let $\F_q$ denote the finite field with $q$ elements, where $q$ is a prime power. For any $\textbf x\in \F_q^N$, the Hamming weight of $\textbf x$ is the number of nonzero components of $\textbf x$. An $[N,K,D]$ {\em linear code} $C$ over $\F_q$ is a $K$-dimensional subspace of $\F_q^N$, where $D$ is the minimum nonzero Hamming weight of $C$. Let $A_i$ denote the number of codewords with Hamming weight $i$ in $C$, where $0\leq i\leq N$. The sequence $(A_0,A_1, \ldots,A_N)$ is called the {\em weight distribution} of $C$, and $\sum_{i=0}^NA_{i}z^i$ is called the {\em weight enumerator polynomial} of $C$. It is well-known that the weight enumerator polynomials contain important information that allows the computation of the error probability for some algorithms \cite{Klove}. In the past 70 years, one problem of considerable interest has been to find the weight enumerator polynomials of a family of linear codes with good parameters. For example, MDS codes \cite{MacWilliams} and some cyclic codes \cite{cyclic-1}. 
In general, it is very hard to determine the weight enumerator polynomials of an infinite family of optimal linear codes and the reader is referred to \cite{cyclic-3,Hu-1}.

Let $C$ be an $[N,K]$ linear code over $\F_q$. For any subcode $D$ of $C$, the {\em support weight} of $D$ is defined by
$$\chi(D):=\#\{i~|~{\rm there~ is}~(c_1,c_2,\ldots,c_N)\in D~{\rm such~ that}~ c_i\neq 0\}.$$
If $D$ has dimension 1, then $\chi(D)$ is exactly the Hamming weight of any nonzero codeword of $D$.
For $1\leq r\leq K$ and $0\leq i\leq N$, let $A_i^{(r)}$ denote the number of $r$-dimensional subcodes of $C$ with support weight $i$. The sequence $(A_0^{(r)},A_1^{(r)}, \ldots,A_N^{(r)})$ is called the $r$-th {\em support weight distribution} of $C$. 
The $r$-th {\em minimum support weight} of $C$ is defined by
$$\dd_r(C):=\min\{\chi(D)~|~D~{\rm is~an~}[N,r]~{\rm subcode~of}~C\}=\min\{i~|~A_i^{(r)}\neq 0\}.$$
In fact, the concept of support weight distribution can be traced back to the 1970s, which was originally proposed by Helleseth, Kl{\o}ve, and Mykkeltveit \cite{Tor-3,KloveDM}.
The support weight distribution has application in describing the security of information data when a code is used in the wire-tap channel of type II \cite{Wei-IT}. Besides, it is also used to determine the separation property and the trellis complexity of linear codes \cite{1993IT-Kasami,1994IT-Forney}.

In 1991, Wei \cite{Wei-IT} studied the $r$-th minimum support weight and called it the $r$-th {\em generalized Hamming weight}. Additionally, Wei also mentioned the notation of the {\em weight hierarchy}, $i.e.$, the set of integers $\{\dd_r(C)~|~1\leq i\leq K\}$, and showed that the weight hierarchy of $C$ determined the performance of $C$ on wire-tap channel of Type II. Subsequently, a series of excellent conclusions on the generalized Hamming weight were provided in \cite{Tor-6,Shi-IT-2023,KloveDM1992}. However, it may often be extremely hard to determine the weight hierarchies and usually even harder to determine the support weight distribution. To the best of our knowledge, there are many papers on the generalized Hamming weight, but there are few known results on the support weight distribution.

Let $C_1$ be an $[N,K]$ linear code over $\F_q$ with generator matrix $G$. For $\ell\geq 1$, one can consider the lifted code $C_\ell$ with the same generator matrix $G$ over $\F_{q^\ell}$. 
Helleseth, Kl{\o}ve, and Mykkeltveit \cite{Tor-3} proved an interesting fact that the weight enumerator polynomials of the codes $C_\ell$ are related in a nice way and can be found by considering the generator matrix $G$ only.
They reduced the problem of determining the weight enumerator polynomial of a linear code to the problem of determining the number of certain subspaces, i.e.,
$$A_\ell(z)=1+\sum_{i=1}^N\sum_{j=1}^KA_{i}^{(j)}(q^\ell-1)(q^\ell-q)\cdots (q^\ell-q^{j-1})z^i,$$
where $A_{i}^{(j)}$ is the number of $(K-j)$-dimensional subspaces of $\F_q^K$ which contain exactly $N-i$ of the $N$ columns of $G$. This result was also discovered independently by Greene \cite{Greene}. Later, Helleseth, Kl{\o}ve, and Ytrehus \cite{Tor-6} showed that $A_{i}^{(j)}$ is the number of $j$-dimensional subcodes of $C$ with support weight $i$. 
Therefore, the determination of the support weight distribution of $C_1$ is equivalent to the determination of the weight distribution of the lifted code $C_\ell$ and also to the determination of the number of certain subspaces.


Although the problem of the support weight distribution has been transformed into a pure mathematical problem, it is still quite difficult and challenging.
The support weight distributions of some irreducible cyclic codes were obtained in \cite{Tor-3}. In \cite{Tor-2}, Helleseth considered the support weight distributions of two classes of codes. One is related to Simplex codes and the other is related to MDS codes. The main motivation is that these codes have remarkable combinatorial properties. Later, Helleseth \cite{Tor-1} proved a similar result for the weight distribution of the coset leaders. Recently, Luo and Liu \cite{2023IT-LuoLiu} determined the support weight distributions for several classes of optimal linear codes which were constructed by using the down-sets.

The Solomon-Stiffler codes are a family of famous Griesmer codes by puncturing certain coordinates of the Simplex codes, which were proposed by Solomon and Stiffler \cite{SS-code}. They also presented encoding and decoding procedures for the Solomon-Stiffler codes. 
Subsequently, a series of work related to the Solomon-Stiffler codes were presented \cite{Hamda-1,Tor-5,Belov-1,Tor-4}.
The reason why the Solomon-Stiffler codes have been widely concerned is not only their optimality, but also their geometric applications (see \cite{Ding-1,Ding-2}).
In this paper, we describe the definition of the Solomon-Stiffler codes in the language of subspaces. We determine the weight enumerator polynomials of the lifted codes of the projective Solomon-Stiffler codes. As a result, we determine the support weight distributions of the projective Solomon-Stiffler codes. In particular, we determine the weight hierarchies of the projective Solomon-Stiffler codes. Many non-trivial examples are presented to verify the correctness of our results and have been carefully verified by MAGMA \cite{magma}.

The paper is organized as follows. In Section 2, we give some notation and definitions.
In Sections 3 and 4, we determine the weight enumerator polynomials of the lifted codes of two families of the Solomon-Stiffler codes using some combinatorial properties of subspaces. In Section 5, we conclude the paper.

\section{Preliminaries}

\subsection{Linear codes}
Let $C$ be an $[N, K, D]$ linear code over $\F_q$. The code $C$ is {\em projective} if the minimum nonzero weight of its dual code is at least three.
Two codes $C$ and $D$ are {\em equivalent} if there is a monomial matrix $M$ such
that $MC = D$.
 The dual code $C^{\perp}$ of $C$ is defined by
$$C^{\perp}=\{\textbf y\in \F_q^N~|~\langle \textbf x, \textbf y\rangle=0, {\rm for\ all}\ \textbf x\in C \},$$
where $\langle \textbf x, \textbf y\rangle=\sum_{i=1}^N x_iy_i$ for $\textbf x = (x_1,x_2, \ldots, x_N)$ and $\textbf y = (y_1,y_2, \ldots, y_N)\in \F_q^N$.
The {\em Singleton bound} \cite[Chapter 11, Section 1]{MacWilliams} for $C$ is defined as $D \leq N-K +1.$ The code $C$ is called a {\em maximum distance separable (MDS) code} if $D=N-K+1$. The {\em Griesmer bound} (see \cite{Griesmer} or \cite[Chapter 17, Section 5]{MacWilliams}) for $C$ is given by
$$N\geq \sum_{i=0}^{K-1}\left\lceil \frac{D}{q^i}\right\rceil,$$
where $\lceil a\rceil$ is the least integer greater than or equal to $a$.
The code $C$ is called a {\em Griesmer code} if its parameters achieve the Griesmer bound with equality.

\subsection{The projective Solomon-Stiffler codes}

Let $S_{q,u,k}$ denote a $u$-dimensional subspace of $\F_q^k$, so $S_{q,k,k}=\F_q^k$.
 Let $S_{q,u,k}^1$ be a set whose elements are made up of precisely one nonzero vector from each vector subspace of dimension $1$ of $S_{q,u,k}$. Without loss of generality, suppose that $S_{q,u,k}^1\subseteq S_{q,k,k}^1$. Let $\left[S_{q,u,k}^1\right]$ be a matrix whose columns are elements of $S_{q,u,k}^1$.
Then the code with the generator matrix $\left[S_{q,k,k}^1\right]$ is called a {\em Simplex code} over $\F_q$ with parameters $\left[ \frac{q^k-1}{q-1},k,q^{k-1} \right]$ and weight enumerator $1+(q^k-1)z^{q^{k-1}}$. The dual code of a Simplex code is called a {\em Hamming code}, which has parameters $\left[ \frac{q^k-1}{q-1},\frac{q^k-1}{q-1}-k,3 \right]$.
Up to monomial equivalence, Simplex codes and Hamming codes are unique, and are denoted by $\mathcal{S}_{q,k}$ and $\mathcal{H}_{q,k}$, respectively.

For $1\leq i\leq p$, let $u_i$ be a positive integer. If
$$
\sum_{i=1}^pu_i\leq k,~1\leq u_i\leq k-1,~u_i\neq u_j,~S_{q,u_i,k}\cap S_{q,u_j,k}=\{0\}~(i\neq j),
$$
then the linear code $SS_1$ with the following generator matrix
$$SS_{q,k}=\left[S_{q,k,k}^1\setminus \bigcup_{i=1}^pS_{q,u_i,k}^1\right]$$
is called the {\em projective Solomon-Stiffler code}~\cite{SS-code}. According to \cite[Theorem 2']{SS-code}, the projective Solomon-Stiffler code is a Griesmer code with parameters
$$\left[\frac{q^k-1}{q-1}-\sum_{i=1}^p\frac{q^{u_i}-1}{q-1},k,
q^{k-1}-\sum_{i=1}^pq^{u_i-1}\right].$$
For convenience, consider the code $SS_2$ with the following generator matrix
$$[SS_{q,k},\alpha SS_{q,k},\ldots,\alpha^{q-2} SS_{q,k}],$$
where $\alpha$ is a primitive element of $\F_q$. We call $SS_1$ a {\em projective code} of $SS_2$.
Let $W_i(z)$ denote the weight enumerator polynomial of $SS_i~(i=1,2)$. Then $W_2(z)=W_1(z^{q-1})$.

\subsection{Some known results}

For a number $q$ with $q\neq 1$, the Gaussian or $q$-binomial
coefficient $\begin{bmatrix}
                     n \\
                     k
                   \end{bmatrix}$ is defined to be
$$\begin{bmatrix}
                    n \\
                    0
                  \end{bmatrix}=1~{\rm and}~\begin{bmatrix}
                    n \\
                    k
                  \end{bmatrix}=
\frac{(q^n-1)(q^n-q)\cdots(q^n-q^{k-1})}{(q^k-1)(q^k-q)\cdots(q^k-q^{k-1})}.$$
The Gaussian coefficients have the same symmetry as that of
binomial coefficients, $i.e.,$
$\begin{bmatrix}
                     n \\
                     k
                   \end{bmatrix}
=\begin{bmatrix}
                     n \\
                     n-k
                   \end{bmatrix}.$
The number of $k$-dimensional subspaces of $\F_q^n$ is just $\begin{bmatrix}
                     n \\
                     k
                   \end{bmatrix}$.

The following theorem which connects the enumerator polynomial $A_\ell(z)$ of the various $V_\ell$ was proved by Helleseth, Kl{\o}ve, and Mykkeltveit in 1977 \cite{Tor-3}.

\begin{theorem}{\rm\cite{Tor-3}}\label{thm-A_l}
Let $C_1$ be an $[n,k]$ linear code over $\F_{q}$. Then the lifted code $C_\ell$ over $\F_{q^\ell}$ of $C_1$ has the following weight enumerator polynomial
$$A_\ell(z)=1+\sum_{i=1}^n\sum_{j=1}^kA_{i}^{(j)}(q^\ell-1)(q^\ell-q)\cdots (q^\ell-q^{j-1})z^i,$$
where $A_{i}^{(j)}$ is the number of $(k-j)$-dimensional subspaces of $\F_q^k$ which contain exactly $n-i$ of the $n$ columns of $G$.
\end{theorem}

Although Helleseth et al. \cite{Tor-6} mentioned in the end that the $q$-ary case is similar to the binary case, it is not redundant to reprove it again. For the completeness of the results, we provide the following proof. For any ${\bf x}\in \F_q^k$, the multiplicity $\mu_G({\bf x})$ of ${\bf x}$ in $G$ is defined as the number of occurrences of ${\bf x}$ as a column vector in $G$.

\begin{theorem}{\rm\cite{Tor-6}}\label{support-weight}
Let $C$ be an $[n,k]$ linear code over $\F_q$ with generator matrix $G$. Then the number of $j$-dimensional subcodes of $C$ with support weight $i$ is $A_{i}^{(j)}$, i.e., the number of $(k-j)$-dimensional subspaces of $\F_q^k$ which contain exactly $n-i$ of the $n$ columns of $G$.
\end{theorem}

\begin{proof}
If $A$ is a $j\times k$ matrix of rank $j$, then the matrix $AG$ generates an $[n,j]$ subcode $D$ of $C$, and it is obvious that any $[n,j]$ subcode of $C$ is obtained by this way. It follows that
$$\chi(D)=n-\sum_{A{\bf x}={\bf 0}}\mu_G({\bf x}).$$
Let $U$ be the space orthogonal to the column space of $A$. Then
$$\sum_{A{\bf x}={\bf 0}}\mu_G({\bf x})=\sum_{{\bf x}\in U}\mu_G({\bf x}).$$
Since the arbitrariness of $A$, $U$ can traverse all $(k-j)$-dimensional subspaces of $\F_q^k$. If $D$ has the support weight $\chi(D)=i$, i.e., $\sum_{{\bf x}\in U}\mu_G({\bf x})=n-i$. This implies that the number of $j$-dimensional subcodes of $C$ with support weight $i$ is the number of $(k-j)$-dimensional subspaces of $\F_q^k$ which contain exactly $n-i$ of the $n$ columns of $G$.
\end{proof}

A simple application is to consider the case of Simplex codes, also shown in \cite{Tor-3}.

\begin{theorem}
Let $\mathcal{S}_{q,k}$ be the $[n=\frac{q^k-1}{q-1},k,q^{k-1}]$ simplex code over $\F_q$. Then the lifted code over $\F_{q^\ell}$ of $\mathcal{S}_{q,k}$ has the following weight enumerator polynomial
$$A_\ell(z)=1+\sum_{j=1}^k\begin{bmatrix}
                    k \\
                   j
                  \end{bmatrix}
(q^\ell-1)(q^\ell-q)\cdots (q^\ell-q^{j-1})z^{\frac{q^k-q^{k-j}}{q-1}}.$$
In particular, $\mathcal{S}_{q,k}$ has the $j$-th support weight distribution as follows.
\begin{center}
\setlength{\tabcolsep}{10.5mm}
\begin{tabular}{cc}
\hline
   Support weight& Frequency \\
    \hline\hline
    $\frac{q^k-q^{k-j}}{q-1}$ &$\begin{bmatrix}
                    k \\
                    j
                  \end{bmatrix}$ \\
    \hline
\end{tabular}
\end{center}
Moreover, $\mathcal{S}_{q,k}$ has the weight hierarchy
$\{\dd_1(\mathcal{S}_{q,k}),\dd_2(\mathcal{S}_{q,k}),\ldots,\dd_k(\mathcal{S}_{q,k})\},$
where $\dd_j(\mathcal{S}_{q,k})=\frac{q^k-q^{k-j}}{q-1}$ for $1\leq j\leq k$.
\end{theorem}

\begin{proof}
Let $\mathcal{G}$ be a set whose elements are the columns of a generator matrix $[S^1_{q,k,k}]$ of $\mathcal{S}_{q,k}$.
Assume that $V$ is a $(k-j)$-dimensional subspace of $\F_q^k$.
Then
\begin{center}
$|V\cap \mathcal{G}|=|V\setminus \{{\bf 0}\}|=\frac{q^{k-j}-1}{q-1}$.
\end{center}
By Theorem \ref{thm-A_l}, $A_{i}^{(j)}$ is the number of $(k-j)$-dimensional subspaces of $\F_q^k$ and contain exactly $n-i$ of the $n$ columns of $[S^1_{q,k,k}]$.
So $A_{i}^{(j)}=0$ for $n-i\neq \frac{q^{k-j}-1}{q-1}$.
When $n-i=\frac{q^{k-j}-1}{q-1}$, i.e., $i=\frac{q^k-q^{k-j}}{q-1}$, $A_{i}^{(j)}$ is the number of $(k-j)$-dimensional subspaces of $\F_q^k$. So
$$A_{i}^{(j)}=\begin{bmatrix}
                    k \\
                    k-j
                  \end{bmatrix}=\begin{bmatrix}
                    k \\
                    j
                  \end{bmatrix},$$
where $i=\frac{q^k-q^{k-j}}{q-1}$. Hence
\begin{align*}
  A_\ell(z)= & 1+\sum_{i=1}^n\sum_{j=1}^kA_{i}^{(j)}(q^\ell-1)(q^\ell-q)\cdots (q^\ell-q^{j-1})z^i \\
  = & 1+\sum_{j=1}^k  \begin{bmatrix}
                    k \\
                    j
                  \end{bmatrix}(q^\ell-1)(q^\ell-q)\cdots (q^\ell-q^{j-1})z^{\frac{q^k-q^{k-j}}{q-1}}.
\end{align*}
In particular, the $j$-th support weight distribution of $C$ is obtained by Theorem \ref{support-weight}. It is well-known that the $j$-th generalized Hamming
weight of a code is defined as the minimum support weight of its $j$-dimensional subcodes. This completes the proof.
\end{proof}

\section{The weight enumerator polynomials of the lifted codes of the Solomon-Stiffler codes $(p=1)$}

In this section, we use $S_{u}$ to denote a $u$-dimensional subspace of $\F_q^k$. We study  the weight enumerator polynomials of the lifted codes of a family of Solomon-Stiffler codes.

\begin{lem}\label{lem-subspace-u}
 Let $0\leq t\leq u\leq k-1$ and $N(k,u,j,t)$ denote the number of $j$-dimensional subspaces $V$ of $\F_q^k$ such that $\dim(V\cap S_{u})=t$. Then
\begin{center}
$N(k,u,j,t)=
q^{(u-t)(j-t)}\begin{bmatrix}
                    k\!-\!u \\
                    j\!-\!t
                  \end{bmatrix}
  \begin{bmatrix}
                    u \\
                    t
                  \end{bmatrix}, $
\end{center}
\end{lem}

\begin{proof}
Let $V$ be a $j$-dimensional subspace of $S_{k}=\F_q^k$.
Assume that $V\cap S_{u}=V_2$ is a $t$-dimensional subspace. Then there is a $(j-t)$-dimensional subspace $V_1\subset S_k\backslash S_u$ such that
$$V=V_1\oplus V_2.$$

Let $\{v_1,v_2,\ldots,v_{j-t}\}$ denote a basis of $V_1$. For $v_1$, there are $q^k-q^u$ choices. For $v_2$, there are $q^k-q^{u+1}$ choices. For $v_i$, there are $q^k-q^{u+i-1}$ choices. Hence, there are $(q^k-q^u)(q^k-q^{u+1})\cdots(q^k-q^{u+j-t-1})$ distinct ordered choices for $(v_1,v_2,\ldots,v_{j-t})$.
For any $(j-t)$-dimensional subspace $V_1$, there are $(q^{j-t}-1)(q^{j-t}-q)\cdots(q^{j-t}-q^{j-t-1})$ choices for $(v_1,v_2,\ldots,v_{j-t})$ such that $\{v_1,v_2,\ldots,v_{j-t}\}$ is a basis of $V_1$. Therefore,
the number of $(j-t)$-dimensional subspaces $V_1$ such that $V_1\subset S_k\backslash S_u$ is $$\frac{(q^k-q^u)(q^{k}-q^{u+1})\cdots (q^k-q^{u+j-t-1})}{(q^{j-t}-1)(q^{j-t}-q)\cdots(q^{j-t}-q^{j-t-1})}=
q^{u(j-t)} \begin{bmatrix}
                    k\!-\!u \\
                    j\!-\!t
                  \end{bmatrix}.$$
The number of $t$-dimensional subspaces $V_2$ of $\F_q^k$ such that $V_2\subset S_u$ is $\begin{bmatrix}
                    u\\
                    t
                  \end{bmatrix}.$
Similarly, for any $V$ and $V_2$, the number of $(j-t)$-dimensional subspaces $V_1$ such that $V=V_1\oplus V_2$
is
$$\frac{(q^j-q^t)(q^{j}-q^{t+1})\cdots (q^j-q^{j-1})}{(q^{j-t}-1)(q^{j-t}-q)\cdots(q^{j-t}-q^{j-t-1})}=
q^{t(j-t)}.$$
 Hence,
 \begin{align*}
   N(k,u,j,t) &= |\{V\subseteq \F_q^k~|~\dim(V)=j,~\dim(V\cap S_u)=t\}| \\
   &=\frac{|\{V=V_1\oplus V_2~|~\dim(V_1)=j-t,~\dim(V_2)=t,~V_1\subset S_k\backslash S_u,~V_2\subset S_u\}|}{q^{t(j-t)}}\\
   &=\frac{|\{V_1~|~\dim(V_1)=j-t,~V_1\subset S_k\backslash S_u\}|\cdot
   |\{V_2~|~\dim(V_2)=t,~V_2\subset S_u\}|}{q^{t(j-t)}}\\
   &= q^{(u-t)(j-t)}\begin{bmatrix}
                    k\!-\!u \\
                    j\!-\!t
                  \end{bmatrix}
\begin{bmatrix}
                    u\\
                    t
                  \end{bmatrix}.
 \end{align*}
 This completes the proof.
\end{proof}

\begin{cor}\label{cor-1}
Let $t\leq u\leq k-1$ and $N(k,u,j,t)$ denote the number of $j$-dimensional subspaces $V$ of $\F_q^k$ such that $\dim(V\cap S_u)=t$. Then
$$\sum_{t=0}^uN(k,u,j,t)=\sum_{t=0}^u
q^{(u-t)(j-t)}\begin{bmatrix}
                    k\!-\!u \\
                    j\!-\!t
                  \end{bmatrix}
 \begin{bmatrix}
                    u \\
                    t
                  \end{bmatrix}
=  \begin{bmatrix}
                    k \\
                    j
                  \end{bmatrix}. $$
In particular, when $u=1$, we have
$$q^j \begin{bmatrix}
                    k\!-\!1 \\
                    j
                  \end{bmatrix}+  \begin{bmatrix}
                    k\!-\!1 \\
                    j\!-\!1
                  \end{bmatrix}
= \begin{bmatrix}
                    k \\
                    j
                  \end{bmatrix}.$$
\end{cor}

\begin{proof}
Let $\mathfrak{F}_{k,u,j,t}$ denote a set consisting of $j$-dimensional subspaces $V$ of $\F_q^k$ satisfying $\dim(V\cap S_u)=t$. Then
$$|\mathfrak{F}_{k,u,j,t}|=N(k,u,j,t)=q^{(u-t)(j-t)} \begin{bmatrix}
                    k\!-\!u \\
                    j\!-\!t
                  \end{bmatrix}  \begin{bmatrix}
                    u \\
                    t
                  \end{bmatrix}$$
 and $\mathfrak{F}_{k,u,j,0},\mathfrak{F}_{k,u,j,1},\ldots,\mathfrak{F}_{k,u,j,u}$ form a partition of all $j$-dimensional subspaces of $\F_q^k$.
 Hence
 $$\sum_{t=0}^uN(k,u,j,t)=\sum_{t=0}^u
q^{(u-t)(j-t)} \begin{bmatrix}
                    k\!-\!u \\
                    j\!-\!t
                  \end{bmatrix} \begin{bmatrix}
                    u \\
                    t
                  \end{bmatrix}
 =  \begin{bmatrix}
                    k \\
                    j
                  \end{bmatrix}.$$
Substitute $u=1$ into the above formula, we can obtain the desired result.
\end{proof}

In the following, let $[U\setminus V]$ denote a matrix whose columns consist of the elements of $U$ which do not contain the elements of $V$, where $U$ and $V$ are two sets whose elements are vectors of length $k$.

\begin{theorem}\label{thm-puncture-u}
Suppose that $u\leq k-1$. Let $C$ be the $q$-ary linear code with the generator matrix $[S_k\backslash S_u]$.
Then the lifted code over $\F_q^\ell$ of $C$ is a linear code of length $n=q^k-q^u$ with the following weight enumerator polynomial
$$A_\ell(z)=1+\sum_{i=1}^n\sum_{j=1}^kA_{i}^{(j)}(q^\ell-1)(q^\ell-q)\cdots (q^\ell-q^{j-1})z^i,$$
where
\begin{center}
$A_{i}^{(j)}$=$\left\{
\begin{array}{ll}
q^{(u-t)(k-j-t)} \begin{bmatrix}
                    k\!-\!u \\
                    k\!-\!j\!-\!t
                  \end{bmatrix}
 \begin{bmatrix}
                    u \\
                    t
                  \end{bmatrix}, & {\rm if}\ n-i=q^{k-j}-q^t,\ \max\{0,u-j\}\leq t\leq u, \\
 0, & {\rm otherwise}.
 \end{array}
 \right.$
\end{center}
\end{theorem}

\begin{proof}
Let $\mathcal{G}$ be a set whose elements are the columns of $G$.
Assume that $V$ is a $(k-j)$-dimensional subspace of $\F_q^k$.
If $\dim(V\cap S_u)=t$, then
\begin{center}
$|V\cap \mathcal{G}|=|V\setminus (V\cap S_u)|=q^{k-j}-q^t$, where $0\leq t \leq u$.
\end{center}
By Theorem \ref{thm-A_l}, $A_{i}^{(j)}$ is the number of $(k-j)$-dimensional subspaces of $\F_q^k$ which contain exactly $n-i$ of the $n$ columns of $G$.
So $A_{i}^{(j)}=0$ for $n-i\neq q^{k-j}-q^t$.
When $n-i=q^{k-j}-q^t$, $A_{i}^{(j)}$ is the number of $(k-j)$-dimensional subspaces $V$ of $\F_q^k$ such that $\dim(V\cap S_u)=t$.
By Lemma \ref{lem-subspace-u},
$$A_{i}^{(j)}=N(k,u,k-j,t)=q^{(u-t)(k-j-t)} \begin{bmatrix}
                    k\!-\!u \\
                    k\!-\!j\!-\!t
                  \end{bmatrix}
  \begin{bmatrix}
                    u \\
                    t
                  \end{bmatrix}.$$
Note that $\dim(V\cap S_u)\leq \dim(S_u)\leq u$ and
 \begin{align*}
    \dim(V\cap S_u)& =\dim(V)+\dim(S_u)-\dim(V\oplus S_u)\geq \max\{0,u-j\}.
\end{align*}
This completes the proof.
\end{proof}

\begin{remark}
When $\ell=1$, the Solomon-Stiffler codes in Theorem \ref{thm-puncture-u} are a family of
two-weight codes, which is also obtained in \cite[Proposition 1]{Hu-1}.
When $\ell=1$ and $u=1$, the projective codes of the Solomon-Stiffler codes in Theorem \ref{thm-puncture-u} are a family of two-weight codes, which is also obtained in \cite[Theorem 18]{Liu-1}.
\end{remark}

\begin{cor}
Let $C$ be the linear code in Theorem $\ref{thm-puncture-u}$. Then $C$ has the $j$-th support weight distribution
$(A_0^{(j)},A_1^{(j)}, \ldots,A_N^{(j)}),$
where
\begin{center}
$A_{i}^{(j)}$=$\left\{
\begin{array}{ll}
q^{(u-t)(k-j-t)} \begin{bmatrix}
                    k\!-\!u \\
                    k\!-\!j\!-\!t
                  \end{bmatrix}
 \begin{bmatrix}
                    u \\
                    t
                  \end{bmatrix}, & {\rm if}\ n-i=q^{k-j}-q^t,\ \max\{0,u-j\}\leq t\leq u, \\
 0, & {\rm otherwise}.
 \end{array}
 \right.$
\end{center}
In particular, $C$ has the weight hierarchy
$\{\dd_1(C),\dd_2(C),\ldots,\dd_k(C)\},$
where
$$\dd_j(C)=\left\{\begin{array}{ll}
                    q^k-q^u-q^{k-j}+q^{u-j}, & {\rm if}~1\leq j< u,\vspace{1mm} \\
                    q^k-q^u-q^{k-j}+1, & {\rm if}~u\leq j\leq k.
                  \end{array}
\right.$$
\end{cor}

\begin{proof}
The $j$-th support weight distribution of $C$ is obtained by Theorem \ref{support-weight} and Theorem \ref{thm-puncture-u}. It is well-known that the $j$-th generalized Hamming weight of a code is defined as the minimum support weight of its $j$-dimensional subcodes. Hence when $1\leq j\leq u$, we have $\max\{0,u-j\}=u-j$. It follows that
$$\dd_j(C)=\min_{u-j\leq t\leq u}\{q^k-q^u-q^{k-j}+q^t\}= q^k-q^u-q^{k-j}+q^{u-j}.$$
When $j>u$, we have $\max\{0,u-j\}=0$. It follows that
$$\dd_j(C)=\min_{0\leq t\leq u}\{q^k-q^u-q^{k-j}+q^t\}= q^k-q^u-q^{k-j}+1.$$
This completes the proof.
\end{proof}

\begin{example}
We consider the linear code $C$ over $\F_{3}$ with the generator matrix $G=[S_3\backslash S_2]$, where $S_2=\langle 110,011\rangle$ is a $2$-dimensional subspace of $S_3=\F_3^3$. Then $C$ is an $[18,3,12]$ linear code with the following generator matrix
$$\left[\begin{array}{c|c}
  1   0   0   1   2   0   1   2   1  & 2   0   0   2   1   0   2   1   2 \\
  0   1   0   2   1   2   0   2   1   &0   2   0   1   2   1   0   1   2 \\
  0   0   1   2   0   1   1   1   1  & 0   0   2   1   0   2   2   2   2
\end{array}\right].
$$
By Theorem \ref{thm-puncture-u},
\begin{center}
$A_{12}^{(1)}=12,$
 $A_{18}^{(1)}=1,$
 $A_{16}^{(2)}=9,$
 $A_{18}^{(2)}=4,$
 $A_{18}^{(3)}=1,$
 $A_{i}^{(j)}=0$ otherwise.
\end{center}
Note that $(3^\ell-1)(3^\ell-3)\cdots(3^\ell-3^{j-1})=0$ when $j\geq 3$. Hence the lifted code of $C$ over $\F_{3^2}$ has the weight enumerator polynomial
$$A_2(z)=1+96z^{12}+432z^{16}+200z^{18}.$$
Moreover, $C$ has the $j$-th support weight distribution as follows.
\begin{center}
\setlength{\tabcolsep}{1mm}
\begin{tabular}{cc||cc||cc}
\hline
\multicolumn{2}{c}{{\rm Case 1: $j=1$}}
&\multicolumn{2}{c}{{\rm Case 2: $j=2$}}
&\multicolumn{2}{c}{{\rm Case 3: $j=3$}}\\
\hline
   Support weight& Frequency&Support weight& Frequency&Support weight& Frequency \\
    \hline
    12 &12 &16 &9&18 &1 \\
    18 &1 &18 &4& &  \\
    \hline
\end{tabular}
\end{center}
In particular, $C$ has the weight hierarchy
$\{12,16,18\}.$
Consider the projective code $C'$ of $C$ with the following generator matrix
$$\left[\begin{array}{c}
  1   0   0   1   2   0   1   2   1   \\
  0   1   0   2   1   2   0   2   1    \\
  0   0   1   2   0   1   1   1   1
\end{array}\right].
$$
This is a projective $[9,3,6]$ linear code, whose the lifted code over $\F_{3^2}$ has the weight enumerator polynomial
$$A'_2(z)=1+96z^{6}+432z^{8}+200z^{9}.$$
Moreover, $C'$ has the $j$-th support weight distribution as follows.
\begin{center}
\setlength{\tabcolsep}{1mm}
\begin{tabular}{cc||cc||cc}
\hline
\multicolumn{2}{c}{{\rm Case 1: $j=1$}}
&\multicolumn{2}{c}{{\rm Case 2: $j=2$}}
&\multicolumn{2}{c}{{\rm Case 3: $j=3$}}\\
\hline
   Support weight& Frequency&Support weight& Frequency&Support weight& Frequency \\
    \hline
    6 &12 &8 &9&9 &1 \\
    9 &1 &9 &4& &  \\
    \hline
\end{tabular}
\end{center}
In particular, $C'$ has the weight hierarchy
$\{6,8,9\}.$
This is also confirmed by MAGMA \cite{magma}.
\end{example}

\section{The weight enumerator polynomials of the lifted codes of the Solomon-Stiffler codes $(p=2)$}
We will consider the case of $p=2$. The result shows that the weight enumerator polynomials of the Solomon-Stiffler codes $(p\geq 3)$ will be too complicated and therefore less interesting.

\begin{lem}\label{lem-t}
Let $S_{u_1}$ be a $u_1$-dimensional subspace of $\F_q^k$ and $S_{u_2}$ be a $u_2$-dimensional subspace of $\F_q^k$ such that $\dim(S_{u_1}\cap S_{u_2})=0$. Let $N_t(k,u_1,u_2,t_1,t_2)$ denote the number of $(t+t_1+t_2)$-dimensional subspaces $V$ of $S_{u_1}\oplus S_{u_2}$ such that $\dim(V\cap S_{u_1})=t_1$ and $\dim (V\cap S_{u_2})=t_2$. Then
$$N_t(k,u_1,u_2,t_1,t_2)=\left\{\begin{array}{ll}
(q^{t}-1)\cdots(q^{t}-q^{t-1})\begin{bmatrix}
                    u_1\!-\!t_1\\
                    t
                  \end{bmatrix}
  \begin{bmatrix}
                    u_2\!-\!t_2\\
                    t
                  \end{bmatrix}
 \begin{bmatrix}
                    u_1\\
                    t_1
                  \end{bmatrix}
 \begin{bmatrix}
                    u_2\\
                    t_2
                  \end{bmatrix}, & t\neq 0, \\
                            \begin{bmatrix}
                    u_1\\
                    t_1
                  \end{bmatrix}
 \begin{bmatrix}
                    u_2\\
                    t_2
                  \end{bmatrix}, & t=0.
                         \end{array}
                         \right.
$$
\end{lem}

\begin{proof}
Assume that $V$ is a $(t+t_1+t_2)$-dimensional subspace of $S_{u_1}\oplus S_{u_2}$ such that $\dim(V \cap S_{u_1})=t_1$ and $\dim(V\cap S_{u_2})=t_2$. Then it is easy to check that $V$ can be written in the following form
$$V=S_{t_1}\oplus S_{t_2}\oplus V_{t},$$
where $S_{t_r}$ is $t_r$-dimensional subspace of $S_{u_r}$ for $1\leq r\leq 2$, $V_{t}$ is a $t$-dimensional subspace of $S_{u_1}\oplus S_{u_2}$ such that $V_{t}\backslash\{0\}\subset (S_{u_1}\backslash S_{t_1})\oplus (S_{u_2}\backslash S_{t_2})$.
For $1\leq r\leq 2$, the number of $t_r$-dimensional subspaces $S_{t_r}$ of $S_{u_r}$ is
$\begin{bmatrix}
                    u_r\\
                    t_r
                  \end{bmatrix}.$
If $t=0$, then the number of $(t_1+t_2)$-dimensional subspaces $V$ of $S_{u_1}\oplus S_{u_2}$ such that $\dim(V\cap S_{u_1})=t_1$ and $\dim(V\cap S_{u_2})=t_2$ is
$$N_0(k,u_1,u_2,t_1,t_2)= \begin{bmatrix}
                    u_1\\
                    t_1
                  \end{bmatrix}
  \begin{bmatrix}
                    u_2\\
                    t_2
                  \end{bmatrix}.$$

Assume that $t\neq 0$. For any $v_1+v_2, v_1'+v_2'\in (S_{u_1}\backslash S_{t_1})\oplus (S_{u_2}\backslash S_{t_2})$,
it is not difficult to check that
$$S_{t_1}\oplus S_{t_2}\oplus \langle v_1+v_2\rangle=S_{t_1}\oplus S_{t_2}\oplus \langle v_1'+v_2'\rangle$$
 if and only if $v_r$ and $v_r'$ are in the same coset of $S_{u_r}$ in $S_{t_r}$ for $1\leq r\leq 2$. Therefore, for the elements in the same coset, we only need to consider one representative element.

Let $\{v_1,v_2,\ldots,v_t\}$ denote a basis of $V_t$. For $v_1$, there are $(q^{u_1-t_1}-1)(q^{u_2-t_2}-1)$ choices. For $v_2$, there are $(q^{u_1-t_1}-q)(q^{u_2-t_2}-q)$ choices. For $v_i$, there are $(q^{u_1-t_1}-q^{i-1})(q^{u_2-t_2}-q^{i-1})$ choices. Hence, there are  $(q^{u_1-t_1}-1)(q^{u_2-t_2}-1)\cdots (q^{u_1-t_1}-q^{t-1})(q^{u_2-t_2}-q^{t-1})$ distinct ordered choices for $(v_1,v_2,\ldots,v_t)$.
For any $t$-dimensional subspace $V_t$, there are $(q^{t}-1)(q^{t}-q)\cdots(q^{t}-q^{t-1})$ choices for $(v_1,v_2,\ldots,v_t)$ such that $\{v_1,v_2,\ldots,v_t\}$ is a basis of $V_t$.
Therefore, the number of $t$-dimensional subspaces $V_{t}$ of $S_{u_1}\oplus S_{u_2}$ such that $V_{t}\backslash\{0\}\subset (S_{u_1}\backslash S_{t_1})\oplus (S_{u_2}\backslash S_{t_2})$ is
 \begin{align*}
    & \frac{(q^{u_1-t_1}-1)(q^{u_2-t_2}-1)\cdots (q^{u_1-t_1}-q^{t-1})(q^{u_2-t_2}-q^{t-1})}
{(q^{t}-1)(q^{t}-q)\cdots(q^{t}-q^{t-1})} \\
   = & (q^{t}-1)(q^{t}-q)\cdots(q^{t}-q^{t-1})  \begin{bmatrix}
                    u_1\!-\!t_1\\
                    t
                  \end{bmatrix}
  \begin{bmatrix}
                    u_2\!-\!t_2\\
                    t
                  \end{bmatrix}.
 \end{align*}
According to the above statement, the number of $(t+t_1+t_2)$-dimensional subspaces $V$ of $S_{u_1}\oplus S_{u_2}$ such that $\dim(V\cap S_{u_1})=t_1$ and $\dim (V\cap S_{u_2})=t_2$
is
$$N_t(k,u_1,u_2,t_1,t_2)=(q^{t}-1)\cdots(q^{t}-q^{t-1}) \begin{bmatrix}
                    u_1\!-\!t_1\\
                    t
                  \end{bmatrix}
  \begin{bmatrix}
                    u_2\!-\!t_2\\
                    t
                  \end{bmatrix}
  \begin{bmatrix}
                    u_1\\
                    t_1
                  \end{bmatrix}
  \begin{bmatrix}
                    u_2\\
                    t_2
                  \end{bmatrix}.$$
This completes the proof.
\end{proof}

\begin{remark}
When $t>\min\{u_1-t_1,u_2-t_2\}$, $N_t(k,u_1,u_2,t_1,t_2)=0$, so we suppose that
$0\leq t\leq \min\{u_1-t_1,u_2-t_2\}$.
\end{remark}

\begin{lem}\label{lem-subspace-u-1-2}
Let $S_{u_1}$ be a $u_1$-dimensional subspace of $\F_q^k$ and $S_{u_2}$ be a $u_2$-dimensional subspace of $\F_q^k$ such that $\dim(S_{u_1}\cap S_{u_2})=0$.
Let $N(k,j,u_1,u_2,t_1,t_2)$ denote the number of $j$-dimensional subspaces $V$ of $\F_q^k$ such that $\dim(V\cap S_{u_1})=t_1$ and $\dim(V\cap S_{u_2})=t_2$. Then
$$N(k,j,u_1,u_2,t_1,t_2)=\sum_{t=0}^{t'}q^{(u_1+u_2-t-t_1-t_2)(j-t-t_1-t_2)}N_t(k,u_1,u_2,t_1,t_2)
\begin{bmatrix}
                    k\!-\!u_1\!-\!u_2 \\
                    j\!-\!t\!-\!t_1\!-\!t_2
                  \end{bmatrix},$$
 where $t'=\min\{u_1-t_1,u_2-t_2\}$ and
$$N_t(k,u_1,u_2,t_1,t_2)=\left\{\begin{array}{ll}
(q^{t}-1)\cdots(q^{t}-q^{t-1}) \begin{bmatrix}
                    u_1\!-\!t_1\\
                    t
                  \end{bmatrix}
                  \begin{bmatrix}
                    u_2\!-\!t_2\\
                    t
                  \end{bmatrix}
                  \begin{bmatrix}
                    u_1\\
                    t_1
                  \end{bmatrix}
                  \begin{bmatrix}
                    u_2\\
                    t_2
                  \end{bmatrix}, & t\neq 0, \\
                    \begin{bmatrix}
                    u_1\\
                    t_1
                  \end{bmatrix}
  \begin{bmatrix}
                    u_2\\
                    t_2
                  \end{bmatrix}, & t=0.
                         \end{array}
                         \right.
$$
\end{lem}

\begin{proof}
For any $j$-dimensional subspace $V$ of $\F_q^k$ with $\dim(V\cap S_{u_1})=t_1$ and $\dim(V\cap S_{u_2})=t_2$, it can be written in the following form
$$V=V_1\oplus V_2,$$
where $V_1$ is a $(j-t-t_1-t_2)$-dimensional subspace of $\F_{q}^k$ such that $\dim(V_1\cap( S_{u_1}\oplus S_{u_2}))=0$, and $V_2$ is a $(t+t_1+t_2)$-dimensional subspace of $S_{u_1}\oplus S_{u_2}$ such that $\dim(V_2 \cap S_{u_1})=t_1$ and $\dim(V_2\cap S_{u_2})=t_2$.

By Lemma \ref{lem-subspace-u}, the number of $(j-t-t_1-t_2)$-dimensional subspaces $V_1$ such that $\dim(V_1\cap( S_{u_1}\oplus S_{u_2}))=0$ is
\begin{align*}
   & \frac{(q^k-q^{u_1+u_2})(q^{k}-q^{u_1+u_2+1})\cdots (q^k-q^{u_1+u_2+j-t-t_1-t_2-1})}
{(q^{j-t-t_1-t_2}-1)(q^{j-t-t_1-t_2}-q)\cdots(q^{j-t-t_1-t_2}-q^{j-t-t_1-t_2-1})} \\
   =& q^{(u_1+u_2)(j-t-t_1-t_2)}  \begin{bmatrix}
                    k\!-\!u_1\!-\!u_2 \\
                    j\!-\!t\!-\!t_1\!-\!t_2
                  \end{bmatrix}.
\end{align*}

By Lemma \ref{lem-t}, the number of $(t+t_1+t_2)$-dimensional subspaces $V_2$ of $S_{u_1}\oplus S_{u_2}$ such that $\dim(V_2 \cap S_{u_1})=t_1$ and $\dim(V_2\cap S_{u_2})=t_2$ is
$$N_t(k,u_1,u_2,t_1,t_2)=\left\{\begin{array}{ll}
(q^{t}-1)\cdots(q^{t}-q^{t-1}) \begin{bmatrix}
                    u_1\!-\!t_1\\
                    t
                  \end{bmatrix}
                  \begin{bmatrix}
                    u_2\!-\!t_2\\
                    t
                  \end{bmatrix}
                  \begin{bmatrix}
                    u_1\\
                    t_1
                  \end{bmatrix}
                  \begin{bmatrix}
                    u_2\\
                    t_2
                  \end{bmatrix}, & t\neq 0, \\
                    \begin{bmatrix}
                    u_1\\
                    t_1
                  \end{bmatrix}
  \begin{bmatrix}
                    u_2\\
                    t_2
                  \end{bmatrix}, & t=0.
                         \end{array}
                         \right.
$$
It is easy to see that $N_t(k,u_1,u_2,t_1,t_2)=0$ if $t> \min\{u_1-t_1,u_2-t_2\}$.

By Lemma \ref{lem-subspace-u}, for any given $V$ and $V_2$, the number of $(j-t-t_1-t_2)$-dimensional subspaces $V_1$ such that $V=V_1\oplus V_2$
is
$$\frac{(q^j-q^{t+t_1+t_2})(q^{j}-q^{t+t_1+t_2+1})\cdots (q^j-q^{j-1})}{(q^{j-t-t_1-t_2}-1)(q^{j-t-t_1-t_2}-q)\cdots(q^{j-t-t_1-t_2}-q^{j-t-t_1-t_2-1})}
=q^{(t+t_1+t_2)(j-t-t_1-t_2)}.$$
 Hence
 \begin{align*}
  N(k,j,u_1,u_2,t_1,t_2) &= |\{V\subseteq \F_q^k~|~\dim(V)=j,\dim(V\cap S_{u_1})=t_1,\dim(V\cap S_{u_2})=t_2\}| \\
   &=\sum_{t=0}^{t'}\frac{|\{V_1\oplus V_2~|~V_1\in P_1,V_2\in P_2\}|} {q^{(t+t_1+t_2)(j-t-t_1-t_2)}}\\
   &=\sum_{t=0}^{t'}\frac{|\{V_1~|~V_1\in P_1\}|\cdot
   |\{V_2~|~V_2\in P_2\}|}{q^{(t+t_1+t_2)(j-t-t_1-t_2)}}\\
   &=\sum_{t=0}^{t'}q^{(u_1+u_2-t-t_1-t_2)(j-t-t_1-t_2)}N_t(k,u_1,u_2,t_1,t_2) \begin{bmatrix}
                    k\!-\!u_1\!-\!u_2 \\
                    j\!-\!t\!-\!t_1\!-\!t_2
                  \end{bmatrix},
 \end{align*}
 where \begin{align*}
         P_1= & \left\{V_1\in \F_q^k~|~\dim(V_1)=j-t-t_1-t_2, \dim(V_1\cap(S_{u_1}\oplus S_{u_2}))=0\right\}, \\
         P_2= & \left\{V_2\in (S_{u_1}\oplus S_{u_2})~|~\dim(V_2)=t+t_1+t_2,
 \dim(V_2 \cap S_{u_1})=t_1,\dim(V_2\cap S_{u_2})=t_2\right\}.
       \end{align*}
This completes the proof.
\end{proof}

\begin{theorem}\label{thm-u1-u2}
Let $S_{u_1}$ be a $u_1$-dimensional subspace of $\F_q^k$ and $S_{u_2}$ be a $u_2$-dimensional subspace of $\F_q^k$ such that $\dim(S_{u_1}\cap S_{u_2})=0$.
Let $C$ be the $q$-ary linear code with the generator matrix $G=[S_k\setminus (S_{u_1} \cup S_{u_2})]$.
Then the lifted code over $\F_q^\ell$ of $C$ is a linear code of length $n=q^k-q^{u_1}-q^{u_2}+1$ with the following weight enumerator polynomial
$$A_\ell(z)=1+\sum_{i=1}^n\sum_{j=1}^kA_{i}^{(j)}(q^\ell-1)(q^\ell-q)\cdots (q^\ell-q^{j-1})z^i,$$
where
$$
A_{i}^{(j)}=\left\{
\begin{array}{ll}
  N(k,k\!-\!j,u_1,u_2,t_1,t_2)+N(k,k\!-\!j,u_1,u_2,t_2,t_1),\\
 \ \ \ \ \ \ \ \ \ \ \ \ \ \ \ \ \ \ \ \ \ \ \ \ \ \ \ \ \ \ \ \ \ \ \ \ \ \ \ \ \ \ \ \ \ \ {\rm if}\ n\!-\!i= q^{k-j}\!-\!q^{t_1}\!-\!q^{t_2}\!+\!1,~t_1\neq t_2,\\
  N(k,k\!-\!j,u_1,u_2,t_1,t_2), \\
 \ \ \ \ \ \ \ \ \ \ \ \ \ \ \ \ \ \ \ \ \ \ \ \ \ \ \ \ \ \ \ \ \ \ \ \ \ \ \ \ \ \ \ \ \ \ {\rm if}\ n\!-\!i= q^{k-j}\!-\!q^{t_1}\!-\!q^{t_2}\!+\!1,~t_1=t_2,\\
  0,\ \ \ \ \ \ \ \ \ \ \ \ \ \ \ \ \ \ \ \ \ \ \ \ \ \ \ \ \ \ \ \ \ \ \ \ \ \ \ \ \ \ \ \ \ \ \ \ \ \ \ \ \ \ \ \  \ \ \ \ \ \ \ \ \ \ \ \ \ \ \ \ \ \ {\rm otherwise.}
\end{array}
\right.
$$
 where $N(k,j,u_1,u_2,t_1,t_2)$ is explained in Lemma $\ref{lem-subspace-u-1-2}$ and $\max\{0,u_r-j\}\leq t_r\leq u_r$ for $r=1,2$.
\end{theorem}

\begin{proof}
Let $\mathcal{G}$ be a set whose elements are the columns of $G$.
Assume that $V$ is a $(k-j)$-dimensional subspace of $\F_q^k$.
If $\dim(V\cap S_{u_1})=t_1$ and $\dim(V\cap S_{u_2})=t_2$, then
\begin{center}
$V\cap \mathcal{G}=(q^{k-j}-1)-(q^{t_1}-1)-(q^{t_2}-1)=q^{k-j}-q^{t_1}-q^{t_2}+1$, where $0\leq t_i \leq u_i$.
\end{center}
By Theorem \ref{thm-A_l}, $A_{i}^{(j)}$ is the number of $(k-j)$-dimensional subspaces of $\F_q^k$ which contain exactly $n-i$ of the $n$ columns of $G$.
So $A_{i}^{(j)}=0$ for $n-i\neq q^{k-j}-q^{t_1}-q^{t_2}+1$.
 When $n-i= q^{k-j}-q^{t_1}-q^{t_2}+1$, $A_{i}^{(j)}$ is the number of $(k-j)$-dimensional subspaces $V$ of $\F_q^k$ such that $\dim(V\cap S_{u_1})=a_1$ and $\dim(V\cap S_{u_2})=a_2$, where $(a_1,a_2)\in \{(t_1,t_2),(t_2,t_1)\}$. By Lemma \ref{lem-subspace-u-1-2}, we have the following.
 \begin{itemize}
   \item When $t_1\neq t_2$,
 $A_{i}^{(j)} =N(k,k-j,u_1,u_2,t_1,t_2)+N(k,k-j,u_1,u_2,t_2,t_1),$
   \item When $t_1=t_2$,
  $A_{i}^{(j)} =N(k,k-j,u_1,u_2,t_1,t_2),$
 \end{itemize}
where $N(k,j,u_1,u_2,t_1,t_2)$ was defined in Lemma \ref{lem-subspace-u-1-2}. Note that for $r=1,2$, we have
\begin{align*}
\dim(V\cap S_{u_r})&\leq \dim(S_{u_r})\leq u_r, {\rm and}\\
    \dim(V\cap S_{u_r})& =\dim(V)+\dim(S_{u_r})-\dim(V\oplus S_{u_r})\geq \max\{0,u_r-j\}.
\end{align*}
This completes the proof.
\end{proof}

\begin{cor}
Let $C$ be the linear code in Theorem $\ref{thm-u1-u2}$. Then $C$ has the $j$-th support weight distribution
$(A_0^{(j)},A_1^{(j)}, \ldots,A_N^{(j)}),$
where
$$
A_{i}^{(j)}=\left\{
\begin{array}{ll}
  N(k,k\!-\!j,u_1,u_2,t_1,t_2)+N(k,k\!-\!j,u_1,u_2,t_2,t_1),\\
 \ \ \ \ \ \ \ \ \ \ \ \ \ \ \ \ \ \ \ \ \ \ \ \ \ \ \ \ \ \ \ \ \ \ \ \ \ \ \ \ \ \ \ \ \  {\rm if}\ n\!-\!i= q^{k-j}\!-\!q^{t_1}\!-\!(q^{t_2}\!-\!1),~t_1\neq t_2,\\
  N(k,k\!-\!j,u_1,u_2,t_1,t_2), \\
 \ \ \ \ \ \ \ \ \ \ \ \ \ \ \ \ \ \ \ \ \ \ \ \ \ \ \ \ \ \ \ \ \ \ \ \ \ \ \ \ \ \ \ \ \  {\rm if}\ n\!-\!i= q^{k-j}\!-\!q^{t_1}\!-\!(q^{t_2}\!-\!1),~t_1=t_2,\\
  0,\ \ \ \ \ \ \ \ \ \ \ \ \ \ \ \ \ \ \ \ \ \ \ \ \ \ \ \ \ \ \ \ \ \ \ \ \ \ \ \ \ \ \ \ \ \ \ \ \ \ \ \ \ \ \ \ \ \ \ \ \ \ \ \ \ \ \ \ \ \ \ \ \ \ \ \ {\rm otherwise.}
\end{array}
\right.
$$
 where $N(k,j,u_1,u_2,t_1,t_2)$ is explained in Lemma $\ref{lem-subspace-u-1-2}$.
In particular, $C$ has the weight hierarchy
$\{\dd_1(C),\dd_2(C),\ldots,\dd_k(C)\},$
where
$$\dd_j(C)=\left\{\begin{array}{ll}
                    q^k-q^{u_1}-q^{u_2}-q^{k-j}+q^{u_1-j}+q^{u_2-j}, & {\rm if}~1\leq j<u_1~{\rm and}~1\leq j<u_2,\vspace{1mm} \\
                    q^k-q^{u_1}-q^{u_2}-q^{k-j}+q^{u_1-j}+1, & {\rm if}~1\leq j<u_1~{\rm and}~u_2\leq j\leq k,\vspace{1mm} \\
                    q^k-q^{u_1}-q^{u_2}-q^{k-j}+q^{u_2-j}+1, & {\rm if}~u_1\leq j\leq k~{\rm and}~1\leq j<u_2,\vspace{1mm} \\
                    q^k-q^{u_1}-q^{u_2}-q^{k-j}+2, & {\rm if}~u_1\leq j\leq k~{\rm and}~u_2\leq j\leq k.
                  \end{array}
\right.$$
\end{cor}

\begin{proof}
The $j$-th support weight distribution of $C$ is obtained by Theorem \ref{support-weight} and Theorem \ref{thm-u1-u2}. The $j$-th generalized Hamming weight of a code is defined as the minimum support weight of its $j$-dimensional subcodes. We have the following four cases.
\begin{itemize}
  \item If $1\leq j<u_1~{\rm and}~1\leq j<u_2$, then $\max\{0,u_r-j\}=u_r-j$ for $r=1~{\rm and}~2$. Hence
      \begin{align*}
        \dd_j(C) & =\min_{{\tiny\begin{array}{c}
          u_1-j\leq t_1< u_1 \\
          u_2-j\leq t_2< u_2
        \end{array}}}\{q^k-q^{u_1}-q^{u_2}-q^{k-j}+q^{t_1}+q^{t_2}\} \\
        & = q^k-q^{u_1}-q^{u_2}-q^{k-j}+q^{u_1-j}+q^{u_2-j}.
      \end{align*}
  \item When $1\leq j<u_1~{\rm and}~u_2\leq j\leq k$, we have $\max\{0,u_1-j\}=u_1-j$ and $\max\{0,u_2-j\}=0$. It follows that
      \begin{align*}
        \dd_j(C) & =\min_{{\tiny\begin{array}{c}
          u_1-j\leq t_1< u_1 \\
          0\leq t_2< u_2
        \end{array}}}\{q^k-q^{u_1}-q^{u_2}-q^{k-j}+q^{t_1}+q^{t_2}\} \\
         & = q^k-q^{u_1}-q^{u_2}-q^{k-j}+q^{u_1-j}+1.
      \end{align*}
  \item When $u_1\leq j\leq k~{\rm and}~1\leq j<u_2$, we have $\max\{0,u_1-j\}=0$ and $\max\{0,u_2-j\}=u_2-j$. It follows that
\begin{align*}
        \dd_j(C) & =\min_{{\tiny\begin{array}{c}
          0\leq t_1< u_1 \\
          u_2-j\leq t_2< u_2
        \end{array}}}\{q^k-q^{u_1}-q^{u_2}-q^{k-j}+q^{t_1}+q^{t_2}\} \\
         & = q^k-q^{u_1}-q^{u_2}-q^{k-j}+q^{u_2-j}+1.
      \end{align*}
  \item When $u_1\leq j\leq k~{\rm and}~u_2\leq j\leq k$, we have $\max\{0,u_1-j\}=0$ and $\max\{0,u_2-j\}=0$. It follows that
\begin{align*}
        \dd_j(C) & =\min_{{\tiny\begin{array}{c}
          0\leq t_1< u_1 \\
          0\leq t_2< u_2
        \end{array}}}\{q^k-q^{u_1}-q^{u_2}-q^{k-j}+q^{t_1}+q^{t_2}\} \\
         & = q^k-q^{u_1}-q^{u_2}-q^{k-j}+2.
      \end{align*}
\end{itemize}
This completes the proof.
\end{proof}

\begin{example}
We consider the linear code $C$ over $\F_{2}$ with the generator matrix $G=[S_6\setminus (S_2\cup S_3)]$, where $S_2=\langle100000,010000\rangle$ and $S_3=\langle001000,000100,000010\rangle$. Then $C$ is a $[53,6,26]$ linear code.
By Theorem \ref{thm-puncture-u_r},
\begin{center}
$A_{26}^{(1)}=42$, $A_{28}^{(1)}=14$, $A_{30}^{(1)}=6$, $A_{32}^{(1)}=1$, $A_{39}^{(2)}=168$, $A_{40}^{(2)}=252$,
$A_{41}^{(2)}=84$, $A_{42}^{(2)}=133$, $A_{44}^{(2)}=7$, $A_{45}^{(2)}=4$, $A_{46}^{(2)}=3$, $A_{46}^{(3)}=336$,
$A_{47}^{(3)}=714$, $A_{48}^{(3)}=231$, $A_{49}^{(3)}=85$, $A_{50}^{(3)}=28$, $A_{53}^{(3)}=1$, $A_{50}^{(4)}=378$, $A_{51}^{(4)}=244$, $A_{52}^{(4)}=21$, $A_{53}^{(4)}=8$, $A_{52}^{(5)}=53$, $A_{53}^{(5)}=10$, $A_{53}^{(6)}=1$,
$A_{ij}=0$ otherwise.
\end{center}
Note that $(2^\ell-1)(2^\ell-2)\cdots(2^\ell-2^{j-1})=0$ when $j\geq 4$ and $\ell=3$. Hence the lifted code of $C$ over $\F_{2^3}$ has the weight enumerator polynomial
{\small\begin{align*}
  A_3(z)  =&1+294z^{26}+98z^{28}+42z^{30}+7z^{32}+7056z^{39}+10584z^{40}+3528z^{41} +5586z^{42}+294z^{44} \\
   & +168z^{45}+56574z^{46}+119952z^{47}+38808z^{48}+14280z^{49}+4704z^{50}+168z^{53}.
\end{align*}}
Moreover, $C$ has the $j$-th support weight distribution as follows.
\begin{center}
\setlength{\tabcolsep}{1mm}
\begin{tabular}{cc||cc||cc}
\hline
\multicolumn{2}{c}{{\rm Case 1: $j=1$}}
&\multicolumn{2}{c}{{\rm Case 2: $j=2$}}
&\multicolumn{2}{c}{{\rm Case 3: $j=3$}}\\
\hline
   Support weight& Frequency&Support weight& Frequency&Support weight& Frequency \\
    \hline
    26 &42 &39 &168&46 &336 \\
    28 &14 &40 &252&47 &714  \\
    30 &6 &41 &84&48 &231  \\
    32 &1 &42 &133&49 &85  \\
     & &44 &7&50 &28  \\
      & &45 &4&53 &1  \\
       & &46 &3& &  \\
    \hline\hline
   \multicolumn{2}{c}{{\rm Case 4: $j=4$}}
&\multicolumn{2}{c}{{\rm Case 5: $j=5$}}
&\multicolumn{2}{c}{{\rm Case 6: $j=6$}}\\
\hline
Support weight& Frequency&Support weight& Frequency&Support weight& Frequency \\
    \hline
    50 &378 &52&53&53 &1 \\
    51 &244 &53 &10& &  \\
    52 &21 & && &  \\
    53 &8 & && & \\
    \hline
\end{tabular}
\end{center}
In particular, $C$ has the weight hierarchy
$\{26,39,46,50,52,53\}.$
This is also confirmed by MAGMA \cite{magma}.
\end{example}

When $\ell=1$, the $q$-ary Solomon-Stiffler codes are four-weight codes.

\begin{theorem}\label{thm-l=1}
Let $S_{u_1}$ be a $u_1$-dimensional subspace of $\F_q^k$ and $S_{u_2}$ be a $u_2$-dimensional subspace of $\F_q^k$ such that $\dim(S_{u_1}\cap S_{u_2})=0$.
Let $C$ be the $q$-ary Solomon-Stiffler code with the generator matrix $G=[S_k\setminus (S_{u_1} \cup S_{u_2})]$. Then $C$ is an
$[q^k-q^{u_1}-q^{u_2}+1,k,(q-1)(q^{k-1}-q^{u_1-1}-q^{u_2-1})]$ linear code with the weight distribution as follows:
\begin{center}
\setlength{\tabcolsep}{10.5mm}
\begin{tabular}{ll}
\hline
   {\rm Hamming~ weight}& {\rm Frequency} \\
    \hline\hline
    $0$ & $1$ \\
  $q^k-q^{k-1}$ & $q^{k-u_1-u_2}-1$ \\
  $(q-1)(q^{k-1}-q^{u_2-1})$ & $q^{k-u_1}-q^{k-u_1-u_2}$ \\
 $(q-1)(q^{k-1}-q^{u_1-1})$ & $q^{k-u_2}-q^{k-u_1-u_2}$ \\
 $(q-1)(q^{k-1}-q^{u_1-1}-q^{u_2-1})$ & $q^{k-u_1-u_2}(q^{u_1}-1)(q^{u_2}-1)$\\
    \hline
\end{tabular}
\end{center}
\end{theorem}

\begin{proof}
Suppose that $\ell=1$, we only consider $j=1$ since $(q^\ell-1)(q^\ell-q)\cdots (q^\ell-q^{j-1})=0$ for $j\geq 2$. Assume that $V$ is a $(k-1)$-dimensional subspace of $\F_q^k$. Then $\dim(V\cap S_{u_r})=u_r$ or $u_r-1$ for $r=1,2$. When $t_1=u_1~{\rm and}~t_2=u_2$, we have
$i=q^k-q^{k-1}$.
By Theorem \ref{thm-u1-u2},
$$A_{i}^{1}=\begin{bmatrix}
                  k-u_1-u_2 \\
                  k-1-u_1-u_2
                \end{bmatrix}=\frac{q^{k-u_1-u_2}-1}{q-1}.$$
Therefore, the frequency of $\omega=q^{k}-q^{k-1}$ is $q^{k-u_1-u_2}-1$.
The rest can also be obtained directly and we omit it.
\end{proof}

\begin{remark}
When $u_1=u_2$, the Solomon-Stiffler codes in Theorem \ref{thm-l=1} are a family of three-weight codes, which were also obtained in \cite[Proposition 5]{Hu-1}. When $u_1=u_2=1$, according to Section 2, the projective codes of the Solomon-Stiffler codes in Theorem \ref{thm-l=1} are a family of projective three-weight codes, which were also obtained in \cite[Theorem 19]{Liu-1}.
\end{remark}

When $u_1=1$, we have the following theorem.

\begin{theorem}\label{thm-puncture-u_r}
Let $S_1$ be a one-dimensional subspace of $\F_q^k$ and $S_{u}$ be a $u$-dimensional subspace of $\F_q^k$ such that $\dim(S_u\cap S_1)=0$.
Let $C$ be the $q$-ary linear code with the generator matrix $G=[S_k\setminus (S_{u} \cup S_1)]$.
Then the lifted code over $\F_q^\ell$ of $C$ is a linear code of length
$n=q^k-q^{u}-q+1$ with the weight enumerator polynomial
$$A_\ell(z)=1+\sum_{i=1}^n\sum_{j=1}^kA_{i}^{(j)}(q^\ell-1)(q^\ell-q)\cdots (q^\ell-q^{j-1})z^i,$$
where
\begin{center}
$A_{i}^{(j)}$=$\left\{
\begin{array}{ll}
q^{u(k-j-1)}\left(\begin{bmatrix}
                    u\\
                    1
                  \end{bmatrix}+1\right)\begin{bmatrix}
                    k\!-\!u\!-\!1 \\
                    k\!-\!j\!-\!1
                  \end{bmatrix}+q^{(u-1)(k-j-2)}\frac{(q^u-1)(q^{u-1}-1)}{(q-1)}
                  \begin{bmatrix}
                    k\!-\!u\!-\!1 \\
                    k\!-\!j\!-\!2
                  \end{bmatrix}, \\
 \ \ \ \ \ \ \ \ \ \ \ \ \ \ \ \ \ \ \ \ \ \ \ \ \ \ \ \ \ \ \ \ \ \ \ \ \ \ \ \ \ \ \ \ \ \ \ \ \ \ \ \ \ \ \ \ \ \ \ \ \ \ \ \ \ \ \ \ \ \ \ ~~    {\rm if}\ n-i=q^{k-j}-q,\\
 q^{(u-t)(k-j-t-1)}\begin{bmatrix}
                    u\\
                    t
                  \end{bmatrix}\left(q^{u-t}\begin{bmatrix}
                    k\!-\!u \\
                    k\!-\!j\!-\!t
                  \end{bmatrix}-\begin{bmatrix}
                    k\!-\!u\!-\!1 \\
                    k\!-\!j\!-\!t\!-\!1
                 \end{bmatrix}\right),\\
  \ \ \ \ \ \ \ \ \ \ \ \ \ \ \ \ \ \ \ \ \ \ \ \ \ \ \ \ \ \ \ \ \ \ \ \ \ \ \ \ \ \ \ \ \ \ \ \ \ \ \ \ \ \ \ \ \ \ \ \ \ \ \ \  {\rm if}\ n-i= q^{k-j}-q^{t},\ t\neq 1, \\
 q^{(u-t)(k-j-t-1)}\begin{bmatrix}
                    k\!-\!u\!-\!1 \\
                    k\!-\!j\!-\!t\!-\!1
                  \end{bmatrix}\begin{bmatrix}
                    u\\
                    t
                  \end{bmatrix},\\
 \ \ \ \ \ \ \ \ \ \ \ \ \ \ \ \ \ \ \ \ \ \ \ \ \ \ \ \ \ \ \ \ \ \ \ \ \ \ \ \ \ \ \ \ \ \ \ \ \ \ \ \ \ \ {\rm if}\ n-i= q^{k-j}-q^{t}-q+1,\ t\neq 0, \\
 0,\ \ \ \ \ \ \ \ \ \ \ \ \ \ \ \ \ \ \ \ \ \ \ \ \ \ \ \ \ \ \ \ \ \ \ \ \ \ \ \ \ \ \ \ \ \ \ \ \ \ \ \ \ \ \ \ \ \ \ \ \ \ \ \ \ \ \ \ \ \ \ \ \ \ \ \ \  \ \ \ \ \ {\rm otherwise}.
 \end{array}
 \right.$
\end{center}
\end{theorem}

\begin{proof}
It is easy to check that this is a special case of Theorem \ref{thm-u1-u2},
when $u_1=1$ and $u_2=u$. So $0\leq t_1\leq 1$ and $t_2=t$.
By Theorem \ref{thm-u1-u2}, $A_{i}^{(j)}=0$ if $n-i\neq q^{k-j}-q^{t}$ or $q^{k-j}-q^{t}-q+1$.
Hence we have\\

(1) When $t_1=0$ and $t_2=1$, or $t_1=1$ and $t_2=0$, $n-i= q^{k-j}-q$, we have
    {\small   \begin{align*}
        A_{i}^{(j)} & =N(k,k-j,1,u,0,1)+N(k,k-j,1,u,1,0) \\
         & =q^{u(k-j-1)}\left(\begin{bmatrix}
                    u\\
                    1
                  \end{bmatrix}+1\right)\begin{bmatrix}
                    k\!-\!u\!-\!1 \\
                    k\!-\!j\!-\!1
                  \end{bmatrix}+q^{(u-1)(k-j-2)}\frac{(q^u-1)(q^{u-1}-1)}{(q-1)} \begin{bmatrix}
                    k\!-\!u\!-\!1 \\
                    k\!-\!j\!-\!2
                  \end{bmatrix}.
      \end{align*}}

(2) When $t_1=0$ and $t_2\neq 1$, $n-i= q^{k-j}-q^{t}$, we have
    {\small \begin{align*}
       A_{i}^{(j)}  =&N(k,k-j,1,u,0,t) \\
        =& q^{(u-t+1)(k-j-t)}\begin{bmatrix}
                    u\\
                    t
                  \end{bmatrix}\begin{bmatrix}
                    k\!-\!u\!-\!1 \\
                    k\!-\!j\!-\!t
                  \end{bmatrix}+q^{(u-t)(k-j-t-1)}(q-1)\begin{bmatrix}
                    u\!-\!t\\
                    1
                  \end{bmatrix}\begin{bmatrix}
                    u\\
                    t
                  \end{bmatrix}\begin{bmatrix}
                    k\!-\!u\!-\!1 \\
                    k\!-\!j\!-\!t\!-\!1
                  \end{bmatrix}\\
 =&q^{(u-t)(k-j-t-1)}\begin{bmatrix}
                    u\\
                    t
                  \end{bmatrix}\left( q^{u-t+k-j-t}\begin{bmatrix}
                    k\!-\!u\!-\!1 \\
                    k\!-\!j\!-\!t
                  \end{bmatrix}+q^{u-t}\begin{bmatrix}
                    k\!-\!u\!-\!1 \\
                    k\!-\!j\!-\!t\!-\!1
                  \end{bmatrix}- \begin{bmatrix}
                    k\!-\!u\!-\!1 \\
                    k\!-\!j\!-\!t\!-\!1
                  \end{bmatrix} \right)\\
 =&q^{(u-t)(k-j-t-1)}\begin{bmatrix}
                    u\\
                    t
                  \end{bmatrix}\left( q^{u-t}\begin{bmatrix}
                    k\!-\!u \\
                    k\!-\!j\!-\!t
                  \end{bmatrix}-\begin{bmatrix}
                    k\!-\!u\!-\!1 \\
                    k\!-\!j\!-\!t\!-\!1
                  \end{bmatrix}\right).
     \end{align*}}
The last step is based on Corollary \ref{cor-1}.

(3) When $t_1=1$ and $t_2\neq 0$, $n-i=q^{k-j}-q^{t}-q+1$, we have
     \begin{align*}
       A_{i}^{(j)}=N(k,k-j,1,u,1,t) =q^{(u-t)(k-j-t-1)}\begin{bmatrix}
                    u\\
                    t
                  \end{bmatrix}\begin{bmatrix}
                    k\!-\!u\!-\!1 \\
                    k\!-\!j\!-\!t\!-\!1
                  \end{bmatrix}.
     \end{align*}
This completes the proof.
\end{proof}

\begin{example}
We consider the linear code $C$ over $\F_{3}$ with the generator matrix $G=[S_4\setminus (S_1\cup S_2)]$, where $S_4=\F_3^4$, $S_1=\langle2100\rangle$ and $S_2=\langle1100,1010\rangle$. Then $C$ is a four-weight $[70,4,46]$ linear code with the following generator matrix
$$\left[\begin{array}{c|c}
 1 0 0 0 2 0 1 1 2 1 2 0 1 2 0 1 2 0 1 2 0 1 2 0 1 2 0 1 2 0 1 2 0 1 2  &
 2 0 0 0 1 0 2 2 1 2 1 0 2 1 0 2 1 0 2 1 0 2 1 0 2 1 0 2 1 0 2 1 0 2 1\\
 0 1 0 0 0 1 1 2 2 0 0 1 1 1 2 2 2 0 0 0 1 1 1 2 2 2 0 0 0 1 1 1 2 2 2 &
 0 2 0 0 0 2 2 1 1 0 0 2 2 2 1 1 1 0 0 0 2 2 2 1 1 1 0 0 0 2 2 2 1 1 1 \\
 0 0 1 0 1 1 1 1 1 0 0 0 0 0 0 0 0 1 1 1 1 1 1 1 1 1 2 2 2 2 2 2 2 2 2 &
 0 0 2 0 2 2 2 2 2 0 0 0 0 0 0 0 0 2 2 2 2 2 2 2 2 2 1 1 1 1 1 1 1 1 1 \\
 0 0 0 1 0 0 0 0 0 1 1 1 1 1 1 1 1 1 1 1 1 1 1 1 1 1 1 1 1 1 1 1 1 1 1 &
 0 0 0 2 0 0 0 0 0 2 2 2 2 2 2 2 2 2 2 2 2 2 2 2 2 2 2 2 2 2 2 2 2 2 2
\end{array}\right].
$$
By Theorem \ref{thm-puncture-u_r},
\begin{center}
$A_{46}^{(1)}=24$, $A_{48}^{(1)}=12$, $A_{52}^{(1)}=3$, $A_{54}^{(1)}=1$, $A_{62}^{(2)}=72$, $A_{64}^{(2)}=53$, $A_{66}^{(2)}=4$, $A_{70}^{(2)}=1$, $A_{68}^{(3)}=35$, $A_{70}^{(3)}=5$,
$A_{70}^{(4)}=1$,
$A_{i}^{(j)}=0$ otherwise.
\end{center}
Note that $(3^\ell-1)(3^\ell-3)\cdots(3^\ell-3^{j-1})=0$ when $j\geq 3$ and $\ell=2$. Hence the lifted code of $C$ over $\F_{3^2}$ has the weight enumerator polynomial
$$A_{2}(z)=192z^{46}+96z^{48}+24z^{52}+8z^{54}+3456z^{62}+2544z^{64}+192z^{66}+48z^{70}.$$
Moreover, $C$ has the $j$-th support weight distribution as follows.
\begin{center}
\setlength{\tabcolsep}{4mm}
\begin{tabular}{cc||cc}
\hline
\multicolumn{2}{c}{{\rm Case 1: $j=1$}}
&\multicolumn{2}{c}{{\rm Case 2: $j=2$}}\\
\hline
   Support weight& Frequency&Support weight& Frequency \\
    \hline
    46 &24 &62 &72 \\
    48 &12 &64 &53 \\
    52 &3 &66 &4  \\
    54 &1 &70 &1  \\
    \hline
\multicolumn{2}{c}{{\rm Case 3: $j=3$}}
&\multicolumn{2}{c}{{\rm Case 4: $j=4$}}\\
\hline
Support weight& Frequency&Support weight& Frequency\\
\hline
68 &35&70&1\\
70 &5 && \\
\hline
\end{tabular}
\end{center}
In particular, $C$ has the weight hierarchy
$\{46,62,68,70\}.$
Consider the projective code $C'$ of $C$ with the following generator matrix
$$\left[\begin{array}{c}
 1 0 0 0 2 0 1 1 2 1 2 0 1 2 0 1 2 0 1 2 0 1 2 0 1 2 0 1 2 0 1 2 0 1 2  \\
 0 1 0 0 0 1 1 2 2 0 0 1 1 1 2 2 2 0 0 0 1 1 1 2 2 2 0 0 0 1 1 1 2 2 2 \\
 0 0 1 0 1 1 1 1 1 0 0 0 0 0 0 0 0 1 1 1 1 1 1 1 1 1 2 2 2 2 2 2 2 2 2 \\
 0 0 0 1 0 0 0 0 0 1 1 1 1 1 1 1 1 1 1 1 1 1 1 1 1 1 1 1 1 1 1 1 1 1 1
\end{array}\right].
$$
This is a projective $[35,4,23]$ linear code, whose the lifted code over $\F_{3^2}$ has the weight enumerator polynomial
$$A'_{2}(z)=192z^{23}+96z^{24}+24z^{26}+8z^{27}+3456z^{31}+2544z^{32}+192z^{33}+48z^{35}.$$
Moreover, $C'$ has the $j$-th support weight distribution as follows.
\begin{center}
\setlength{\tabcolsep}{4mm}
\begin{tabular}{cc||cc}
\hline
\multicolumn{2}{c}{{\rm Case 1: $j=1$}}
&\multicolumn{2}{c}{{\rm Case 2: $j=2$}}\\
\hline
   Support weight& Frequency&Support weight& Frequency \\
    \hline
    23 &24 &31 &72 \\
    24 &12 &32 &53 \\
    26 &3 &33 &4  \\
    27 &1 &35 &1  \\
    \hline
\multicolumn{2}{c}{{\rm Case 3: $j=3$}}
&\multicolumn{2}{c}{{\rm Case 4: $j=4$}}\\
\hline
Support weight& Frequency&Support weight& Frequency\\
\hline
34 &35&35&1\\
35 &5 && \\
\hline
\end{tabular}
\end{center}
In particular, $C'$ has the weight hierarchy
$\{23,31,34,35\}.$
This is also confirmed by MAGMA \cite{magma}.
\end{example}

\begin{remark}
For the cases of $p\geq 3$, a similar method can be used to determine $A_{i}^{(j)}$, but it is technically more complicated.

By MacWilliams theorem \cite[Chapter 5]{MacWilliams} or \cite{KloveDM1992}, we can determine the weight enumerator polynomials of lifted codes for the dual codes of two classes of the $q$-ary Solomon-Stiffler codes, as well as their support weight distributions and weight hierarchies.
\end{remark}

\section{Conclusion}
In this paper, we have determined the weight enumerator polynomials of the lifted codes of the Solomon-Stiffler codes using some combinatorial properties of subspaces. As a result, we have determined the support weight distributions of the Solomon-Stiffler codes. In particular, we have determined the weight hierarchies of the Solomon-Stiffler codes. Some nontrivial examples are also given. Under certain conditions, we have obtained some three or four weight codes. The results have been carefully verified by MAGMA \cite{magma}\\

\noindent{\bf Conflict of Interest:}
The authors have no conflicts of interest to declare that are relevant to the content of this article.\\

\noindent{\bf Data Deposition Information:}
Our data can be obtained from the authors upon reasonable request.\\

\noindent{\bf Acknowledgement:}
The research of Minjia Shi and Shitao Li was supported by Natural Science Foundation of China (12071001). The research of Tor Helleseth was supported by the Research Council of Norway under grant number 247742/O70.

\end{document}